\documentclass[12pt]{iopart}
\expandafter\let\csname equation*\endcsname\relax
\expandafter\let\csname endequation*\endcsname\relax
\usepackage{amsmath,amssymb}
\usepackage{graphicx,color}
\usepackage{mathtools}
\usepackage{url}
\usepackage{braket}
\usepackage{cite}
\usepackage{relsize}
\usepackage{hyperref}
\usepackage{afterpage}

\usepackage{subfigure}

\usepackage{amsthm}
\newtheorem{theorem}{Theorem}

\newtheorem{lemma}[theorem]{Lemma}

\DeclarePairedDelimiter{\abs}{\lvert}{\rvert}

\begin{document}
\title{Modified Grover operator for quantum amplitude estimation}

\author{Shumpei~Uno$^{1,2}$,
	Yohichi~Suzuki$^{1}$,
	Keigo~Hisanaga$^{3}$,
	Rudy~Raymond$^{1,4}$,
	Tomoki~Tanaka$^{1,5}$,
	Tamiya~Onodera$^{1,4}$,
	Naoki~Yamamoto$^{1,3}$
}

\address{$^1$ Quantum Computing Center, Keio University, Hiyoshi 3-14-1, Kohoku-ku, Yokohama 223-8522, Japan}
\address{$^2$ Mizuho Information \& Research Institute, Inc., 2-3 Kanda-Nishikicho, Chiyoda-ku, Tokyo, 101-8443, Japan}
\address{$^3$ Department of Applied Physics and Physico-Informatics, Keio University, Hiyoshi 3-14-1, Kohoku-ku, Yokohama 223-8522, Japan}
\address{$^4$ IBM Quantum, IBM Research-Tokyo, 19-21 Nihonbashi Hakozaki-cho, Chuo-ku, Tokyo, 103-8510, Japan}
\address{$^5$ Mitsubishi UFJ Financial Group, Inc.\ and MUFG Bank, Ltd., 2-7-1 Marunouchi, Chiyoda-ku, Tokyo 100-8388, Japan}

\ead{yamamoto@appi.keio.ac.jp}

\begin{abstract}
	In this paper, we propose a quantum amplitude estimation method  that uses a modified Grover operator and quadratically improves the estimation accuracy in the ideal case, as in the
	conventional one using the standard Grover operator.
	Under the depolarizing noise, the proposed method can outperform the conventional one in the sense that it can in principle achieve the ultimate estimation accuracy characterized by the quantum Fisher information in the limit of a large number of qubits, while the conventional one cannot achieve the same value of ultimate accuracy.
	In general this superiority requires a sophisticated adaptive measurement, but we numerically demonstrate that the proposed method can outperform the conventional one and approach to the ultimate accuracy, even with a simple non-adaptive measurement strategy.
\end{abstract}
\vspace{2pc}
\submitto{\NJP}

\section{Introduction}

Quantum computers embody a physically feasible computational model which may exceed
the limit of conventional computers, and have been researched vigorously for several decades.
In particular, the amplitude estimation algorithm~\cite{brassard2002quantum,suzuki2020amplitude,aaronson2020quantum,grinko2019iterative,nakaji2020faster,venkateswaran2020quantum} has attracted a lot of attention as a fundamental subroutine of a wide
range of application-oriented quantum algorithms, such as the Monte Carlo integration~\cite{montanaro2015quantum,rebentrost2018quantum,woerner2019quantum,stamatopoulos2019option,martin2019towards,egger2019credit,miyamoto2019reduction} and machine learning tasks\cite{prakash2014quantum,Wiebe2015,Wiebe2016a,Wiebe2016b,Kerenidis2019q,li2019sublinear,miyahara2020quantum}.
However, those quantum amplitude estimation algorithms are assumed to work on ideal
quantum computers, and their performance on noisy computers should be carefully investigated.
In fact, evaluation as well as mitigation of existing quantum algorithms on noisy quantum
computers has become a particularly important research subject, since real quantum computing
devices are publicly available.
The amplitude estimation algorithm is of course one of those focused algorithms~\cite{brown2020quantum,wang2020bayesian,tanaka2020amplitude}.

The amplitude estimation is the problem of estimating the (unknown) amplitude
of a certain fixed target state, in which the parameter we want to calculate
is encoded.
If we naively measure the state $N$ times and use the measurement results to construct
an estimator, then it gives the estimate with the mean squared error scaling
as $O(1/\sqrt{N})$ due to the central limit theorem.
In contrast, the conventional amplitude estimation algorithms~\cite{brassard2002quantum,suzuki2020amplitude,aaronson2020quantum,grinko2019iterative,nakaji2020faster,venkateswaran2020quantum}, in an ideal setting, quadratically improve the estimation error to $O(1/N)$, using the Grover operator (or the amplitude amplification operator in a wider sense) associated with the target state.
Here let us reconsider the role of target state.
In the Grover search problem~\cite{grover1996fast}, the goal is to hit the unknown target state with high probability;
but in the amplitude estimation problem, the goal is not to get the target state, but is to
obtain an information of the unknown amplitude.
Hence there might be some other ``desirable state'' whose hitting probability can be used
to construct a more precise estimator on the original amplitude, than the conventional
Grover-based one.
In particular, following the discussion in the first paragraph, this paper explores a noise-tolerant state as such a desirable state.

Therefore the questions we should answer are as follows.
Is there actually a noise-tolerant state, such that its hitting probability leads to
a better estimator than the case using the target state, in the presence of noise?
A candidate of such a noise-tolerant state would be a product state, e.g., $\Ket{0}_n$, which may be more robust than the (entangled) target state.
Moreover, if it exists, how much does the new estimator reduce the estimation error?
Does it achieve the optimal precision under some condition?
This paper provides affirmative answers to all these questions.

The analysis is based on the $1$-parameter quantum statistical estimation theory~\cite{helstrom1968minimum,holevo2011probabilistic, braunstein1994statistical,fujiwara2001quantum}.
A one-sentence summary of the theorem used in this paper is that, under some reasonable
assumptions, the estimation error is generally lower bounded by the inverse of classical Fisher
information for a fixed measurement, which is further lower bounded by the inverse of quantum
Fisher information without respect to measurement.
Hence the quantum Fisher information gives the ultimate lower bound of estimation error,
although constructing an estimator that achieves the quantum Fisher information generally
requires a sophisticated adaptive measurement strategy, which is difficult to implement
in practice~\cite{paris2009quantum,giovannetti2011advances,kolodynski2014precision,haase2016precision,toth2014quantum,braun2018quantum}.
In this paper, we assume the depolarizing error, which was observed in the real device~\cite{tanaka2020amplitude}.
Then we take the product state $\ket{0}_{n+1}$ as a noise-tolerant state
and show that, like the conventional amplitude amplification (Grover search)
operator $G$, a newly introduced ``transformed amplitude amplification operator''
$Q$ realizes a $2$-dimensional rotation of $\ket{0}_{n+1}$ in the Hilbert
space and eventually leads to alternative estimator.
Then the Fisher information is calculated;
$F_{\textrm{c}, G}(\cdot)$ and $F_{\textrm{q}, G}(\cdot)$ denote the classical and quantum Fisher information
of the state generated by $G$, respectively;
also $F_{\textrm{c}, Q}(\cdot)$ and $F_{\textrm{q}, Q}(\cdot)$ denote the classical and quantum Fisher
information of the state generated by $Q$, respectively.
Below we list the contribution of this paper.

\begin{itemize}
	\item
	      In the ideal case, $F_{\textrm{c}, G}(N_\textrm{q}) = F_{\textrm{q}, G}(N_\textrm{q}) = F_{\textrm{c}, Q}(N_\textrm{q}) = F_{\textrm{q}, Q}(N_\textrm{q}) = 4 N_\textrm{q}^2$,
	      where $N_\textrm{q}$ denotes the number of operations (note that they are all independent to the
	      unknown parameter).
	      That is, there is no difference between $G$- and $Q$-based methods, and in both cases
	      there exists an estimator (e.g., the maximum likelihood estimator) that achieves the ultimate
	      precision.

	\item
	      Under the depolarizing noise,
	      $F_{\textrm{c}, G}(N_\textrm{q})|_{\textrm{env}} \leq F_{\textrm{c}, Q}(N_\textrm{q})|_{\textrm{env}} \leq F_{\textrm{q}, G}(N_\textrm{q}) = F_{\textrm{q}, Q}(N_\textrm{q})$,
	      where $\cdot|_{\textrm{env}}$ denote the upper envelope of the Fisher information as a function of $N_\textrm{q}$.
	      That is, the $Q$-based method using the noise-tolerant state is better than the
	      conventional $G$-based method, in the sense of Fisher information.
	      In particular, $F_{\textrm{c}, Q}(N_\textrm{q})|_{\textrm{env}} \to F_{\textrm{q}, Q}(N_\textrm{q})$ holds and hence the ultimate precision
	      is achievable in the limit of a large number of qubit.

	\item
	      Despite the above-mentioned general difficulty in achieving the quantum Fisher information,
	      we show numerically that, under the depolarizing noise, a simple non-adaptive estimator
	      for $Q$-based method achieves the precision about $1/3$ to $1/2$ of the $G$-based method, which is about $1.2$ to $1.5$ times the ultimate limit $1/F_{\textrm{q}, Q}(N_\textrm{q})$.

\end{itemize}

The rest of this paper is organized as follows.
In Section~\ref{sec:Preliminaries}, we present the conventional $G$-based amplitude estimation
method, and its classical and quantum Fisher information under the depolarizing noise.
We also show there that the $G$-based method can never attain the quantum Fisher
information under the noise.
Next in Section~\ref{sec:AlternativeOperator}, we propose an alternative $Q$-based method using the
noise-tolerant state, and discuss the classical and quantum Fisher information.
Then it is shown that the quantum Fisher information can be achieved in the limit of a large
number of qubits.
In Section~\ref{sec:NumericalSimulation}, we demonstrate by numerical simulation that the
proposed $Q$-based method achieves a better estimation accuracy than the $G$-based method, with
the simple non-adaptive measurement strategy.
Finally, we conclude this paper in Section~\ref{sec:Conclusion}.

\section{Conventional Method}\label{sec:Preliminaries}

The quantum amplitude estimation is the problem of estimating the value
of an unknown parameter $\theta$ of the following state, given a unitary
operator $A$ acting on the $(n+1)$-qubit initial state $\Ket{0}_{n+1}$:
\begin{equation}
	A\Ket{0}_{n+1} = \cos\theta \Ket{\psi_0}_n \Ket{0}+ \sin\theta\Ket{\psi_1}_n\Ket{1},
	\label{eq:InitialStateG}
\end{equation}
where $\Ket{\psi_0}_n $ and $\Ket{\psi_1}_n$ are normalized $n$-qubit states.
Note that the state~\eqref{eq:InitialStateG} has an additional ancilla qubit to represent solutions (``good'' subspace), in contrast to the original Grover search which searches for solutions in $n$-qubit space.
Measuring the ancilla qubit $N_{\textrm{shot}}$ times gives us a simple
estimate with estimation error (the root squared mean error) scaling as
$1/\sqrt{N_{\textrm{shot}}}$ due to the central limit theorem.
The purpose of amplitude estimation algorithms is to reduce the cost to
achieve the same error;
in this paper, the cost of algorithm is counted by the number of queries
on the operator $A$ and its inverse $A^\dagger$, as in the previous
works~\cite{brassard2002quantum,suzuki2020amplitude,aaronson2020quantum,grinko2019iterative,nakaji2020faster,venkateswaran2020quantum}.
This section is devoted to provide an overview of this conventional
quantum-enhanced method for both noiseless and depolarizing-noisy cases,
with particular focus on the classical and quantum Fisher information.
The discussion mainly follows the previous papers~\cite{suzuki2020amplitude,tanaka2020amplitude,jiang2014quantum,yao2014multiple}.

\subsection{Noiseless Case}

The recent amplitude estimation algorithms~\cite{suzuki2020amplitude,aaronson2020quantum,grinko2019iterative,nakaji2020faster,venkateswaran2020quantum}, which do not require the expensive phase estimation, are composed of the following two steps.
The first step is to amplify the amplitude of the target state and perform the  measurement on the state, and the second step is to estimate the amplitude using
the measurement results and a subsequent classical post-processing.

The first step uses the amplitude amplification operator $G$ defined as~\cite{brassard2002quantum}
\begin{align}
	G = A U_0 A^\dagger  U_f,
	\label{eq:OperatorG}
\end{align}
where $U_0$ and $U_f$ are the reflection operators defined as
\begin{align}
	U_0 & = -\mathbf{I}_{n+1} +2 \Ket{0}_{n+1} \Bra{0}_{n+1},          \\
	U_f & = -\mathbf{I}_{n+1} +2 \mathbf{I}_n \otimes \Ket{0} \Bra{0}.
\end{align}
Here, $\mathbf{I}_n$ is the identity operator on the $n$ qubits space.
Via $m$ times application of the operator $G$ on the state~\eqref{eq:InitialStateG},
we obtain
\begin{equation}
	G^m A \Ket{0}_{n+1}
	= \cos\left( (2m+1)\theta\right)\Ket{\psi_0}_n\Ket{0}
	+ \sin\left((2m+1)\theta\right)\Ket{\psi_1}_n\Ket{1}.
\end{equation}
Note now that the number of queries of $A$ and $A^\dagger$ is
$2m+1$.
In the following, we regard this state as a (vector-valued) function
of the general number of queries, $N_\textrm{q}$, instead of $m$; i.e.,
\begin{equation}
	\Ket{\psi_{G}(\theta,N_\textrm{q})}
	\equiv \cos\left(N_\textrm{q}\theta\right)\Ket{\psi_0}_n\Ket{0}
	+ \sin\left(N_\textrm{q}\theta\right)\Ket{\psi_1}_n\Ket{1},
	\label{eq:StateG^m}
\end{equation}
but remember that $\Ket{\psi_{G}(\theta,N_\textrm{q})}$ can take only the odd number in $N_\textrm{q}$.
We then measure the state
$\rho_G(\theta,N_\textrm{q}) = \Ket{\psi_{G}(\theta,N_\textrm{q})}\Bra{\psi_{G}(\theta,N_\textrm{q})}$
via the set of measurement operators $E_{G,0}=\mathbf{I}_{n}\otimes\ket{0}\bra{0}$
and $E_{G,1}=\mathbf{I}_{n}\otimes\ket{1}\bra{1}$, which distinguishes whether
the last bit is $0$ or $1$;
the probability distribution
$p_G(i;\theta,N_\textrm{q})=\textrm{Tr}(\rho_G(\theta,N_\textrm{q})E_{G,i})$ is then calculated as
\begin{equation}
	\begin{split}
		p_G(0;\theta,N_\textrm{q})&=\cos^2(N_\textrm{q}\theta), \\
		p_G(1;\theta,N_\textrm{q})&=\sin^2(N_\textrm{q}\theta).
		\label{eq:ProbabilityPureG}
	\end{split}
\end{equation}

The second step is to post-process the measurement results sampled from the
probability distribution~\eqref{eq:ProbabilityPureG} to estimate the parameter
$\theta$.
If we have $N_{\textrm{shot}}$ independent measurement results, the mean squared
error of any unbiased estimate $\hat{\theta}$ from the measurement results obeys
the following Cram\'er--Rao inequality:
\begin{equation}
	\mathbb{E}\left[
		(\theta-\hat{\theta})^2\right] \geq \frac{1}{N_{\textrm{shot}} F_{\textrm{c},G}(\theta,N_\textrm{q})},
	\label{eq:CCRBPureG}
\end{equation}
where $\mathbb{E}[\cdots]$ represents the expectation of measurement results with respect to the probability distribution~\eqref{eq:ProbabilityPureG}.
The classical Fisher information $F_{\textrm{c},G}(\theta, N_\textrm{q})$ associated
with the probability distribution~\eqref{eq:ProbabilityPureG} is given by
\begin{equation}
	\begin{split}
		F_{\textrm{c},G}(\theta,N_\textrm{q})
		= \mathbb{E}\left[
			\left(
			\frac{\partial}{\partial \theta}\ln p_G(i;\theta,N_\textrm{q}) \right)^2 \right]
		= \sum_{i\in\{0,1\}} \left(
		\frac{\partial}{\partial \theta}\ln p_G(i;\theta,N_\textrm{q}) \right)^2 p_G(i;\theta,N_\textrm{q})
		=4 N_{q}^2.
		\label{eq:CFisherPureG}
	\end{split}
\end{equation}
It is well known that the maximum likelihood estimation is able to asymptotically
attain the lower bound of Cram\'er--Rao inequality in the limit of a large number
of samples, $N_{\textrm{shot}}\to \infty$~\cite{fisher1923xxi}.
In the previous paper~\cite{suzuki2020amplitude}, it was demonstrated by a numerical simulation (with ${N_{\textrm{shot}}=100}$ and  $\sin^2\theta\in\{2/3,1/3,1/6,1/12,1/24,1/48\}$) that the mean squared error obtained by the maximum likelihood estimation well agrees with the Cram\'er--Rao lower bound.
By combining the Fisher information~\eqref{eq:CFisherPureG} with the Cram\'er--Rao lower bound~\eqref{eq:CCRBPureG}, we find that the root squared estimation error scales as $1/(N_{q}\sqrt{N_{\textrm{shot}}})$.
Denoting the total cost as $n=N_\textrm{q} N_{\textrm{shot}}$, the root squared error scales as $1/n$ in the large $N_\textrm{q}$ limit, which is a quadratic improvement over the naive repeated measurement onto the initial state~\eqref{eq:InitialStateG}, whose error scales as $1/\sqrt{N_{\textrm{shot}}}=1/\sqrt{n}$.

Next, we calculate the quantum Fisher information for this amplitude estimation
problem.
Although we obtain the classical Fisher information~\eqref{eq:CFisherPureG}
with the fixed measurement operators, $E_{G,0}$ and $E_{G.1}$, there is an
innumerable degree of freedom in choosing a set of measurement operators, and
other measurements may provide larger classical Fisher information.
However, it is known that any unbiased estimate $\hat{\theta}$ resulting from
an arbitrary measurement (including POVM) obeys the following quantum Cram\'er--Rao
inequality~\cite{helstrom1968minimum,braunstein1994statistical,holevo2011probabilistic}:
\begin{equation}
	\mathbb{E}\left[
		(\theta-\hat{\theta})^2\right]
	\geq \frac{1}{N_{\textrm{shot}} F_{\textrm{q},G}\left(\theta,N_\textrm{q}\right)},\label{eq:QCRBPure_G}
\end{equation}
where $F_{\textrm{q},G} \left(\theta,N_\textrm{q}\right)$ is the quantum Fisher information
defined as
\begin{equation}
	F_{\textrm{q},G} (\theta,N_\textrm{q}) =  \mathrm{Tr}\left( \rho_G(\theta,N_\textrm{q}) L_S\left(\rho_G\left(\theta,N_\textrm{q}\right)\right)^2 \right).
	\label{eq:QuantumFisherInformation}
\end{equation}
Here $L_S(\rho_G(\theta,N_\textrm{q}))$ is the symmetric logarithmic derivative (SLD)
\footnote{There are other definitions of quantum Fisher information, but it is
	known that the SLD gives the tightest bound for one parameter estimation
	problem~\cite{holevo2011probabilistic,nagaoka2005fisher}.},
which is defined as the Hermitian operator satisfying
\begin{equation}
	\dot{\rho}_G(\theta,N_\textrm{q}) = \frac{1}{2}\left( \rho_G(\theta,N_\textrm{q}) L_S(\rho_G(\theta,N_\textrm{q})) + L_S(\rho_G(\theta,N_\textrm{q})) \rho_G(\theta,N_\textrm{q}) \right).
\end{equation}
Here the overdot represents the derivative with respect to $\theta$.
If the output state is a pure state, the quantum Fisher information is explicitly
given by
\begin{equation}
	\begin{split}
		F_{\textrm{q},G}(\theta,N_\textrm{q})
		&= 4\left( \braket{\dot{\psi}_G(\theta,N_\textrm{q})|\dot{\psi}_G(\theta,N_\textrm{q})}-\abs{\braket{\dot{\psi}_G(\theta,N_\textrm{q})|\psi_G(\theta,N_\textrm{q})}}^2 \right) \\
		&= 4 N_\textrm{q}^2.
		\label{eq:QFisherPureG}
	\end{split}
\end{equation}
In our case, this is exactly the same as the classical Fisher information~\eqref{eq:CFisherPureG}.
This means that the measurement operators $E_{G,0}$ and $E_{G,1}$, which
measure only the last qubit, is the optimal measurement in the absence of noise.

It is worthwhile to comment that the amplitude estimation problem described
above can be regarded as the parameter estimation problem embedded in the
unitary operator $G$.
Such a parameter estimation problem is one of the main target of the field
of quantum metrology~\cite{giovannetti2004quantum,giovannetti2006quantum}.
In the quantum metrology, it is well known that the root squared error of the
estimation for a single unitary operation on each of the $N$ independent
initial states is $1/\sqrt{N}$, which is called the standard quantum limit
or simply the shot noise, while the entanglement or multiple sequential
operation quadratically improves the estimation error to $1/N$, which is
called the Heisenberg scaling.
Unfortunately, it is also known that this quadratic improvement is very
fragile and is ruined even with the infinitesimally small
noise~\cite{kolodynski2010phase,knysh2011scaling,escher2011general}.
We encounter the same issue in the amplitude estimation problem as
described below.

\subsection{Depolarizing Noise Case}\label{sec:DepolarizationG}

Here, we describe a case under the depolarizing noise.
As a noise model, we assume that the following depolarizing noise
acts on the system every time when operating $A$ and $A^\dagger$:
\begin{equation}
	\mathcal{D}(\rho) = r \rho + (1-r) \frac{\mathbf{I}_{n+1}} {d},
	\label{eq:DepolarizationChannel}
\end{equation}
where $\rho$ is an arbitrary density matrix, $r$ is a known constant satisfying $r\in(0,1]$, and $d$ is the dimension of the $(n+1)$-qubit system, i.e., $d=2^{n+1}$.
Then, the amplitude amplification operation $G^m$ under the
depolarizing process, with initial state $A\ket{0}_{n+1}$, gives
\begin{equation}
	\begin{split}
		\rho_{G,r}(\theta,N_\textrm{q}) &= r^{ N_\textrm{q} }  \rho_G(\theta,N_\textrm{q}) + (1-r^{N_\textrm{q}}) \frac{\mathbf{I}_{n+1}} {d},
		\label{eq:StateG^mNoise}
	\end{split}
\end{equation}
where $\rho_G(\theta,N_\textrm{q}) = \Ket{\psi_{G}(\theta,N_\textrm{q})}\Bra{\psi_{G}(\theta,N_\textrm{q})}$.
Recall that $\rho_{G,r}(\theta,N_\textrm{q})$ is represented as a
function of general number of queries, $N_\textrm{q}$, yet it takes
$N_\textrm{q}=2m+1$.
Measuring the state~\eqref{eq:StateG^mNoise} using the set of measurement operators $E_{G,0}$ and $E_{G,1}$, the probability distribution $p_{G,r}(i;\theta,N_\textrm{q})=\textrm{Tr}(\rho_{G,r}(\theta,N_\textrm{q})E_{G,i})$ can be written as
\begin{align}
	p_{G,r}(0;\theta,N_\textrm{q}) & =r^{N_\textrm{q}}\cos^2(N_\textrm{q}\theta)+\frac{1-r^{ N_\textrm{q} }}{2},
	\label{eq:ProbabilityNoiseG0}
	\\
	p_{G,r}(1;\theta,N_\textrm{q}) & =r^{N_\textrm{q}} \sin^2(N_\textrm{q}\theta)+\frac{1-r^{N_\textrm{q}}}{2}.
	\label{eq:ProbabilityNoiseG1}
\end{align}
Note that this probability distribution is independent of the dimension of Hilbert space $d$.
Then the classical Fisher information associated with this probability
distribution is given as~\cite{tanaka2020amplitude}
\begin{equation}
	\begin{split}
		F_{\textrm{c},G,r}(\theta,N_\textrm{q})
		&= \sum_{i\in\{0,1\}} \left(
		\frac{\partial}{\partial \theta}\ln p_{G,r}(i;\theta,N_\textrm{q}) \right)^2 p_{G,r}(i; \theta,N_\textrm{q}) \\
		& = \frac{4 {N_\textrm{q}}^2\sin^2(N_\textrm{q}\theta)\cos^2(N_\textrm{q}\theta)r^{2 N_\textrm{q}}}
		{\left(\frac{1}{2}+ \left( \cos^2 (N_\textrm{q} \theta)-\frac{1}{2} \right)r^{N_\textrm{q}}\right)
			\left(\frac{1}{2}+ \left( \sin^2 (N_\textrm{q} \theta)-\frac{1}{2} \right)r^{N_\textrm{q}}\right)}.
	\end{split}
	\label{eq:CFisherNoiseG}
\end{equation}
This classical Fisher information has the following upper envelope
depicted in Figure~\ref{fig:envelope}, which does not depend on the
unknown parameter $\theta$:
\begin{equation}
	F_{\textrm{c},G,r}(N_\textrm{q}) |_{\textrm{env}} = 4 {N_\textrm{q}}^2 r^{2 N_\textrm{q}}.
	\label{eq:CFisherNoiseG_Max}
\end{equation}
From the Cram\'er--Rao inequality, $1/F_{\textrm{c},G,r}(N_\textrm{q}) |_{\textrm{env}}$ gives a lower bound of the estimation error.
This lower bound can be asymptotically attainable by employing an adaptive measurement strategy, which first divides the $N_{\rm shot}$ states into some small fractions and then adaptively tunes the measurements depending on the previous measurement results, as proven in the channel estimation problem~\cite{gill2000state,fujiwara2001quantum,hayashi2005statistical}.
The above expression shows that, even with a very small noise, the Heisenberg scaling is not achievable asymptotically in the large $N_\textrm{q}$ limit, due to the factor of $r^{2N_\textrm{q}}$.
Note that, nonetheless, it is possible to make a constant factor improvement over the naive sampling onto the initial state~\eqref{eq:InitialStateG}, depending on the magnitude of noise.
The Fisher information of the naive sampling, i.e., \eqref{eq:CFisherNoiseG_Max} with $N_\textrm{q}=1$, is about $4$, while the Fisher information~\eqref{eq:CFisherNoiseG_Max} for $N_\textrm{q}\gg 1$ can far exceed the value $4$ if the noise is sufficiently small.
To discuss this constant factor improvement in more detail, we need to consider Fisher information for a finite range of $N_\textrm{q}$ rather than the asymptotic behavior, which is the subject of the following.

\begin{figure}[tbp]
	\centering
	\includegraphics[width=0.8\linewidth]{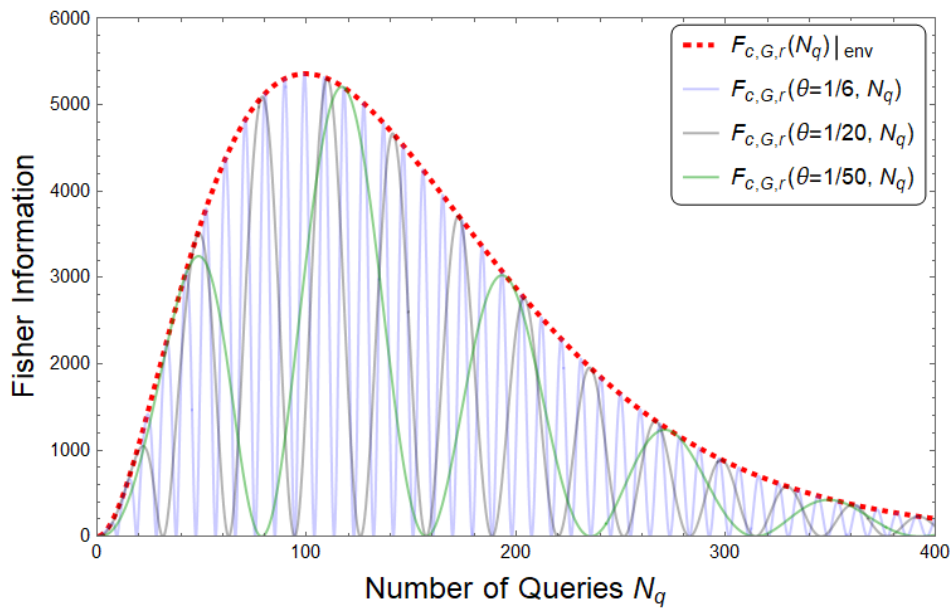}
	\caption{
		Relationship between the number of queries and the classical Fisher information.
		The four solid lines represent the classical Fisher information \eqref{eq:CFisherNoiseG} with several values of $\theta \in \{1/6, 1/20, 1/50\}$.
		The classical Fisher information has the upper envelope \eqref{eq:CFisherNoiseG_Max}, which is depicted by the (red) dotted line.
		The noise parameter is set as $r=0.99$.
	}\label{fig:envelope}
\end{figure}

Next, like the ideal case, let us consider the quantum Fisher information.
As in the previous work on the channel estimation problem~\cite{jiang2014quantum,yao2014multiple},
exploiting the fact that the state of \eqref{eq:StateG^mNoise} can be represented in an exponential form, the SLD operator can be calculated as
\begin{equation}
	L_S(\rho_G(\theta,N_\textrm{q},r)) = \frac{2  r^{N_\textrm{q}}}{\frac{2}{d} + \left(1-\frac{2}{d}\right)r^{N_\textrm{q}}}\dot{\rho}_G(\theta,N_\textrm{q}).
	\label{eq:SLDNoiseG}
\end{equation}
Then the quantum Fisher information can be given as
\begin{equation}
	F_{\textrm{q},G,r}(\theta,N_\textrm{q})= \frac{4 {N_\textrm{q}}^2 r^{2 N_\textrm{q}}}
	{\frac{2}{d}+ \left( 1-\frac{2}{d} \right)r^{N_\textrm{q}} }.
	\label{eq:QFisherNoiseG}
\end{equation}
Importantly, $F_{\textrm{q},G,r}(\theta,N_\textrm{q})$ does not depend on the
unknown parameter $\theta$.
This fact allows us to obtain
$F_{\textrm{q},G,r}(\theta,N_\textrm{q})
	\geq F_{\textrm{c},G,r}(\theta,N_\textrm{q})|_{\textrm{env}}
	\geq F_{\textrm{c},G,r}(\theta,N_\textrm{q})$ for any value of $\theta$.
In the case of single qubit, $n=1$ or equivalently $d=2$, we have
$F_{\textrm{q},G,r}(\theta,N_\textrm{q}) =
	F_{\textrm{c},G,r}(\theta,N_\textrm{q})|_{\textrm{env}}$.
But $F_{\textrm{q},G,r}(\theta,N_\textrm{q})> F_{\textrm{c},G,r}(\theta,N_\textrm{q})|_{\textrm{env}}$ holds when $n \geq 2$.
In particular, in the limit where the qubit number is infinitely large, it follows that
\begin{equation}
	F_{\textrm{q},G,r}(\theta,N_\textrm{q})= 4 {N_\textrm{q}}^2 r^{N_\textrm{q}},
	\label{eq:QFisherNoiseG_InfDim}
\end{equation}
which is bigger than $F_{\textrm{c},G,r}(\theta,N_\textrm{q})|_{\textrm{env}}$
by the factor $r^{N_\textrm{q}}$.
That is, the conventional $G$-based method presented in this section cannot achieve the quantum Fisher information, except the case $n=1$.
In the next section, we present an alternative method that can attain the quantum Fisher information in the limit of a large number of qubits.

It is worthwhile to mention here about the input state of the above method.
The input state $A\Ket{0}_{n+1}$ is an equally weighted superposition of the two eigenstates of the amplitude amplification operator $G$, which is known as the optimal input state in the \textit{sequential strategy} of the channel estimation problem that sequentially applies the unitary $G$ onto single probe~\cite{maccone2013intuitive,yao2014multiple}.
Therefore, the quantum Fisher information~\eqref{eq:QFisherNoiseG} is optimal for any method using the operator $G$.
Furthermore, we prove in~\ref{sec:ProofMaximumQFI} that the quantum Fisher information~\eqref{eq:QFisherNoiseG} is optimal when employing the operator $A$.

\section{Proposed Method}\label{sec:AlternativeOperator}

In this section, we propose an alternative amplitude estimation method, which provides a larger classical Fisher information than that of the
conventional $G$-based method, under the depolarizing noise.
As in the previous section, we first describe the noiseless case and then turn to the case under depolarizing noise.

\subsection{Noiseless Case}

We first define the following modified amplitude amplification operator:
\begin{equation}
	Q =  U_0 A^\dagger U_f A
	\label{eq:OperatorQ}
\end{equation}
This operator is similar to the conventional amplitude amplification operator $G = A U_0 A^\dagger U_f$ given in \eqref{eq:OperatorG}, but there is a big operational meaning as follows.
Note that $G$ induces a rotation by $2\theta$ in the space spanned by
$\ket{\psi_1}_n\ket{1}$ and $\ket{\psi_0}_n\ket{0}$, or equivalently
$A\ket{0}_{n+1}$ and $\ket{\psi_0}_n\ket{0}$.
Then, because of $Q=A^\dagger G A$, we find that $Q$ induces a rotation
by $2\theta$ in the space spanned by $\ket{0}_{n+1}$ and $A^\dagger\ket{\psi_0}_n\ket{0}$ (see Figure~\ref{fig:Rotation}).
Thus, application of $Q^m$ on the initial state $\Ket{0}_{n+1}$ produces
\begin{equation}
	\Ket{\psi_{Q}(\theta,N_\textrm{q})}:=
	Q^m \Ket{0}_{n+1} = \cos(N_\textrm{q}\theta)\Ket{0}_{n+1} + \sin(N_\textrm{q}\theta)\Ket{\phi}_{n+1},
	\label{eq:StateQ^m}
\end{equation}
where
\begin{equation}
	\begin{split}
		\Ket{\phi}_{n+1}
		&=A^\dagger (-\sin\theta \Ket{\psi_0}_n \Ket{0} + \cos \theta\Ket{\psi_1}_n\Ket{1} ) \\
		&=\frac{1}{\sin 2\theta}\left(Q -\cos 2\theta\right) \Ket{0}_{n+1}.
		\label{eq:basisOrthTo0}
	\end{split}
\end{equation}
As in the $G$ case, $\Ket{\psi_{Q}(\theta,N_\textrm{q})}$ is
represented as a (vector-valued) function of the general number of queries, $N_\textrm{q}$, which now takes only the even numbers,
i.e., $N_\textrm{q}=2m$.

\begin{figure}[tbp]
	\centering
	\includegraphics[width=0.8\linewidth]{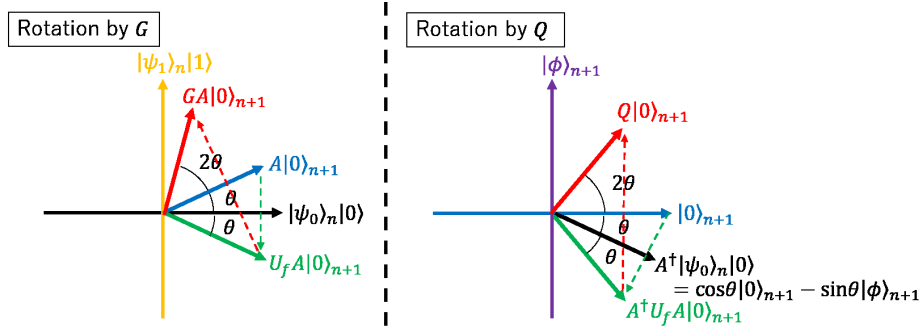}
	\caption{
		The graphical explanation of the action of $G$ (left) and $Q$ (right) on the initial state.
		The operator $G$ defined in \eqref{eq:OperatorG} can be regarded as the rotation with angle $2\theta$ in the space spanned by $\Ket{\psi_0}_n\Ket{0}$ and $\Ket{\psi_1}_n\Ket{1}$.
		First, the action of $U_f$ reflects the initial state $A\Ket{0}_{n+1}$ with respect to $\Ket{\psi_0}_n\Ket{0}$,
		and then the action of $A U_0 A^\dagger$ reflects the resulting state with respect to the initial state $A\ket{0}_{n+1}$,
		so that the rotation with angle $2\theta$ is achieved.
		On the other hand, the operator $Q$ defined in \eqref{eq:OperatorQ} can be regarded as the rotation with angle $2\theta$ in the space spanned by $\Ket{0}_{n+1}$ and $\Ket{\phi}_{n+1}$, i.e., linear transformation via $A$ of the space spanned by $\Ket{\psi_0}_n\Ket{0}$ and $\Ket{\psi_1}_n\Ket{1}$.
		In fact, first $A^\dagger U_f A$ reflects the initial state $\Ket{0}_{n+1}$ with respect to $A^\dagger\Ket{\psi_0}_n\Ket{0}=\cos\theta\Ket{0}_{n+1} -\sin \theta \Ket{\phi}_{n+1}$,
		and then $U_0$ reflects the resulting state with respect to the initial state $\ket{0}_{n+1}$; as a result, $Q\ket{0}_{n+1}$ is
		a $2\theta$-rotated state of $\ket{0}_{n+1}$.
	}\label{fig:Rotation}
\end{figure}

In this method, we employ $E_{Q,0}=\ket{0}_{n+1}\bra{0}_{n+1}$ and $E_{Q,1}=\mathbf{I}_{n+1}-E_{Q,0}$ as the set of measurement operators,
which distinguishes whether all qubits are $0$ or not.
Measuring the state \eqref{eq:StateQ^m} via $\{E_{Q,0}, E_{Q,1}\}$
yields the probability distribution as
\begin{equation}
	\begin{split}
		p_Q(0;\theta,N_\textrm{q})&=\cos^2(N_\textrm{q}\theta), \\
		p_Q(1;\theta,N_\textrm{q})&=\sin^2(N_\textrm{q}\theta).
		\label{eq:ProbabilityQ}
	\end{split}
\end{equation}
The classical Fisher information with respect to this probability
distribution can be calculated as
\begin{align}
	F_{\textrm{c},Q}(\theta,N_\textrm{q}) = 4 N_\textrm{q}^2.
	\label{eq:CFisherPureQ}
\end{align}
This indicates that the estimation error using the operator $Q$
also obeys the Heisenberg scaling in the absence of noise.

The quantum Fisher information for the state~\eqref{eq:StateQ^m} can be calculated in the same manner as \eqref{eq:QFisherPureG};
\begin{align}
	F_{\textrm{q},Q}(\theta,N_\textrm{q}) & = 4 N_\textrm{q}^2.
	\label{eq:QFisherPureQ}
\end{align}
The classical Fisher information~\eqref{eq:CFisherPureQ} coincides with the quantum one \eqref{eq:QFisherPureQ}, while they coincide with the Fisher information~\eqref{eq:CFisherPureG} and~\eqref{eq:QFisherPureG}
for the $G$-based method.
This indicates that the choice of the amplitude amplification operator $G$ or $Q$ does not cause any advantage/disadvantage in the noiseless case.

\subsection{Depolarizing Noise Case}

We assume that the same depolarizing noise~\eqref{eq:DepolarizationChannel} introduced in Section~\ref{sec:DepolarizationG} acts on the system every time when
operating $A$.
Then, the final state after $m$ times operation of $Q$ together with
the noise process, for the initial state $\ket{0}_{n+1}$, is given by
\begin{equation}
	\label{}
	\begin{split}
		\rho_{Q,r}(\theta,N_\textrm{q}) &= r^{ N_\textrm{q} }  \rho_Q(\theta,N_\textrm{q}) + (1-r^{N_\textrm{q}}) \frac{\mathbf{I}_{n+1}} {d},
		\label{eq:StateQ^mNoise}
	\end{split}
\end{equation}
where $\rho_Q(\theta,N_\textrm{q}) = \Ket{\psi_{Q}(\theta,N_\textrm{q})}
	\bra{\psi_{Q}(\theta,N_\textrm{q})}$.
Here, we again use the notation $N_\textrm{q} = 2m$. Measuring this state using the set of measurement operators $E_{Q,0}$ and $E_{Q,1}$ yields the probability distribution:
\begin{align}
	p_{Q,r}(0;\theta,N_\textrm{q}) & =r^{N_\textrm{q}}\cos^2(N_\textrm{q}\theta)+(1-r^{ N_\textrm{q} })\frac{1}{d},
	\label{eq:ProbabilityNoiseQ0}                                                                                    \\
	p_{Q,r}(1;\theta,N_\textrm{q}) & =r^{N_\textrm{q}} \sin^2(N_\textrm{q}\theta)+(1-r^{N_\textrm{q}})\frac{d-1}{d}.
	\label{eq:ProbabilityNoiseQ1}
\end{align}
The classical Fisher information associated with this probability distribution is calculated as
\begin{equation}
	F_{\textrm{c},Q,r}(\theta,N_\textrm{q}) = \frac{4 {N_\textrm{q}}^2\sin^2(N_\textrm{q}\theta)\cos^2(N_\textrm{q}\theta)r^{2 N_\textrm{q}}}
	{\left(\frac{1}{d}+ \left( \cos^2 (N_\textrm{q} \theta)-\frac{1}{d} \right)r^{N_\textrm{q}}\right)
		\left(\frac{d-1}{d}+ \left( \sin^2 (N_\textrm{q} \theta)-\frac{d-1}{d} \right)r^{N_\textrm{q}}\right)},
	\label{eq:CFisherNoiseQ}
\end{equation}
which has the following upper envelope that does not depend on
the unknown parameter $\theta$:
\begin{equation}
	\begin{split}
		F_{\textrm{c},Q,r}(N_\textrm{q}) |_{\textrm{env}}
		&= 4 N_\textrm{q}^2 r^{N_\textrm{q}} +
		8 {N_\textrm{q}}^2\frac{d-1}{d^2}  \left(1-r^{N_\textrm{q}}\right)^2  \\
		&- \frac{8 N_\textrm{q}^2\left(1-r^{N_\textrm{q}}\right)\sqrt{(d-1) \left(d-1+r^{N_\textrm{q}}\right) \left((d-1) r^{N_\textrm{q}}+1\right)}}{d^2}.
		\label{eq:CFisherNoiseQ_Max}
	\end{split}
\end{equation}
Also the quantum Fisher information for the state~\eqref{eq:StateQ^mNoise} can be calculated as
\begin{equation}
	F_{\textrm{q},Q,r}(\theta,N_\textrm{q})= \frac{4 {N_\textrm{q}}^2 r^{2 N_\textrm{q}}}
	{\left(\frac{2}{d}+ \left( 1-\frac{2}{d} \right)r^{N_\textrm{q}}\right) }.
	\label{eq:QFisherNoiseQ}
\end{equation}
Like the case of $G$, this does not depend on $\theta$, which leads
that, because $F_{\textrm{c},Q,r}(\theta,N_\textrm{q})|_{\textrm{env}}$ is
the maximization of $F_{\textrm{c},Q,r}(\theta,N_\textrm{q})$ with respect
to $\theta$, we have
$F_{\textrm{q},Q,r}(\theta,N_\textrm{q})
	\geq F_{\textrm{c},Q,r}(\theta,N_\textrm{q})|_{\textrm{env}}
	\geq F_{\textrm{c},Q,r}(\theta,N_\textrm{q})$.

Now we compare
\footnote{
	Note that while we compare here as if each method takes all integer values, the actual number of queries $N_\textrm{q}$ only takes even and odd values for the $G$-based and $Q$-based methods, respectively.
}
the Fisher information of $G$ and $Q$.
The first notable fact is that their quantum Fisher information are
identical, i.e.,
$F_{\textrm{q},G,r}(\theta,N_\textrm{q})=F_{\textrm{q},Q,r}(\theta,N_\textrm{q})$.
Hence a difference may appear for the values of classical Fisher information, especially their $\theta$-independent envelopes $F_{\textrm{c},Q,r}(N_\textrm{q})|_{\textrm{env}}$
and $F_{\textrm{c},G,r}(N_\textrm{q})|_{\textrm{env}}$.
Actually, using the AM-GM inequality, we can prove
\begin{align*}
	F_{\textrm{c},Q,r}(N_\textrm{q}) |_{\textrm{env}}
	 & \geq 4 N_\textrm{q}^2 r^{N_\textrm{q}} +
	8 {N_\textrm{q}}^2\frac{d-1}{d^2}  \left(1-r^{N_\textrm{q}}\right)^2 \\
	 & - \frac{8 N_\textrm{q}^2\left(1-r^{N_\textrm{q}}\right)}{d^2}
	\left\{
	\frac{ \left(d-1+r^{N_\textrm{q}}\right) +(d-1)\left((d-1) r^{N_\textrm{q}}+1\right)}{2}
	\right\}
	\\ \nonumber
	 & = 4 {N_\textrm{q}}^2 r^{2 N_\textrm{q}}
	= F_{\textrm{c},G,r}(N_\textrm{q}) |_{\textrm{env}}.
\end{align*}
The equality holds only when $d=2$.
Summarizing, we have
\begin{equation}
	F_{\textrm{c},G,r}(N_\textrm{q}) |_{\textrm{env}}
	\leq F_{\textrm{c},Q,r}(N_\textrm{q}) |_{\textrm{env}}
	\leq F_{\textrm{q},Q,r}(\theta,N_\textrm{q})
	= F_{\textrm{q},G,r}(\theta,N_\textrm{q}).
\end{equation}
Now considering the fact that $F_{\textrm{c},G,r}(N_\textrm{q}) |_{\textrm{env}}$ never achieve its ultimate bound $F_{\textrm{q},G,r}(\theta,N_\textrm{q})$ except
the case $d=2$, we are concerned with the achievability in the $Q$ case.
Actually $F_{\textrm{c},Q,r}(N_\textrm{q}) |_{\textrm{env}}= F_{\textrm{q},Q,r}(\theta,N_\textrm{q})$ holds when $d=2$.
But more importantly, it can be proven that
\begin{equation}
	\lim_{d\to \infty}\Big\{
	F_{\textrm{c},Q,r}(N_\textrm{q}) |_{\textrm{env}} -  F_{\textrm{q},Q,r}(\theta,N_\textrm{q}) \Big\} = 0.
	\nonumber
\end{equation}
This implies that, for a large number of qubits, the enveloped classical
Fisher information nearly achieves the ultimate bound.
The intuitive understanding for this remarkable difference between $Q$
and $G$ is that the pure noise effect for the proposed $Q$-based method, the second term of \eqref{eq:ProbabilityNoiseQ0}, exponentially decreases with respect to the number of qubits, which would allow us to efficiently extract the signal component given by the first term of \eqref{eq:ProbabilityNoiseQ0}.
On the other hand, the pure noise effect on the $G$-based method,
the second term of both of \eqref{eq:ProbabilityNoiseG0} and~\eqref{eq:ProbabilityNoiseG1} is about $1/2$, and thus the signal may
be easily buried in the noise.

\begin{figure}[tbp]
	\begin{minipage}[b]{0.5\linewidth}
		\begin{minipage}[b]{\linewidth}
			\centering
			\subfigure[$n=1\ (d = 2)$]{\includegraphics[width=\linewidth]{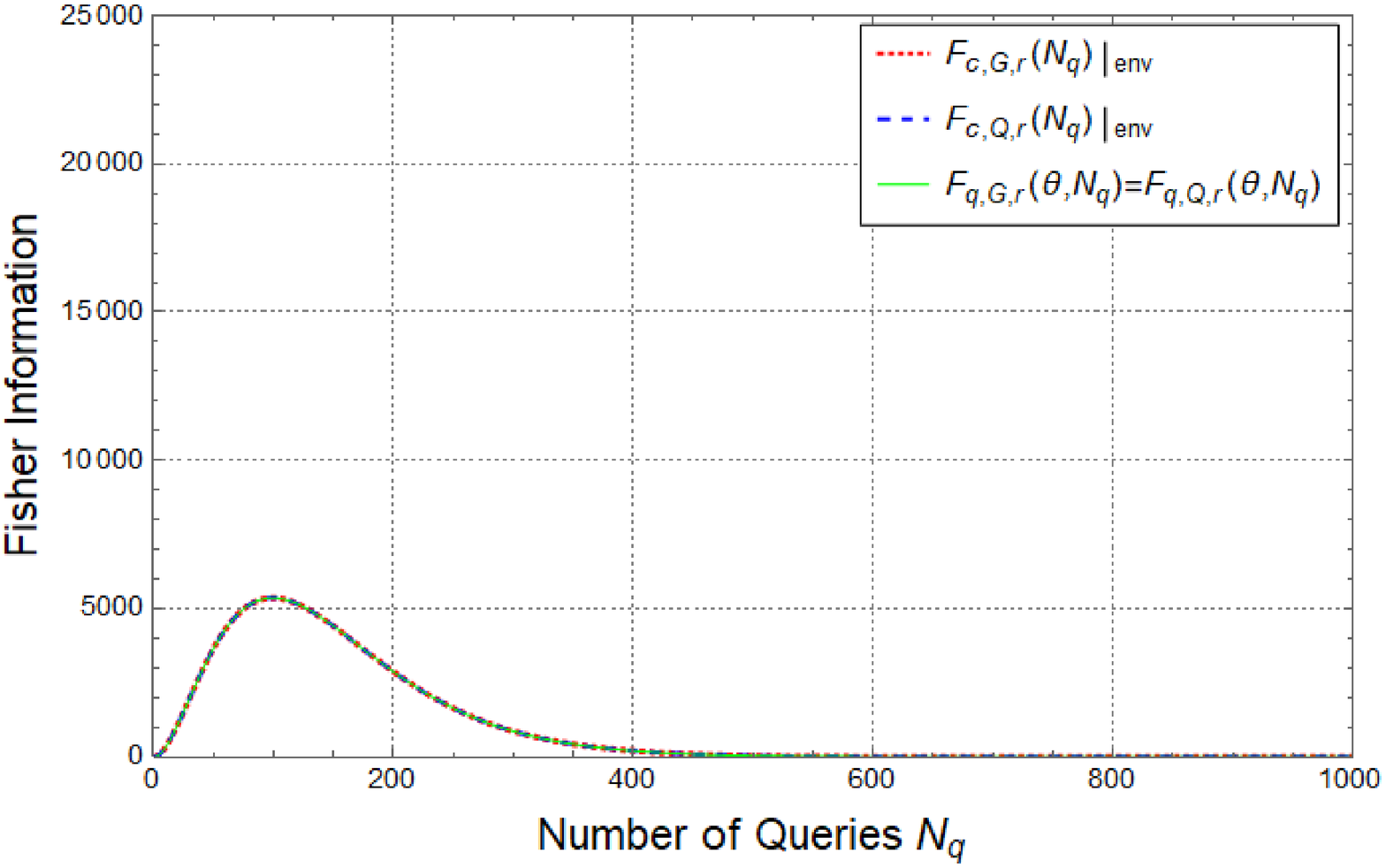}}
		\end{minipage}
		\\\\\\
		\begin{minipage}[b]{\linewidth}
			\centering
			\subfigure[$n=10\ (d = 2^{10})$]{\includegraphics[width=\linewidth]{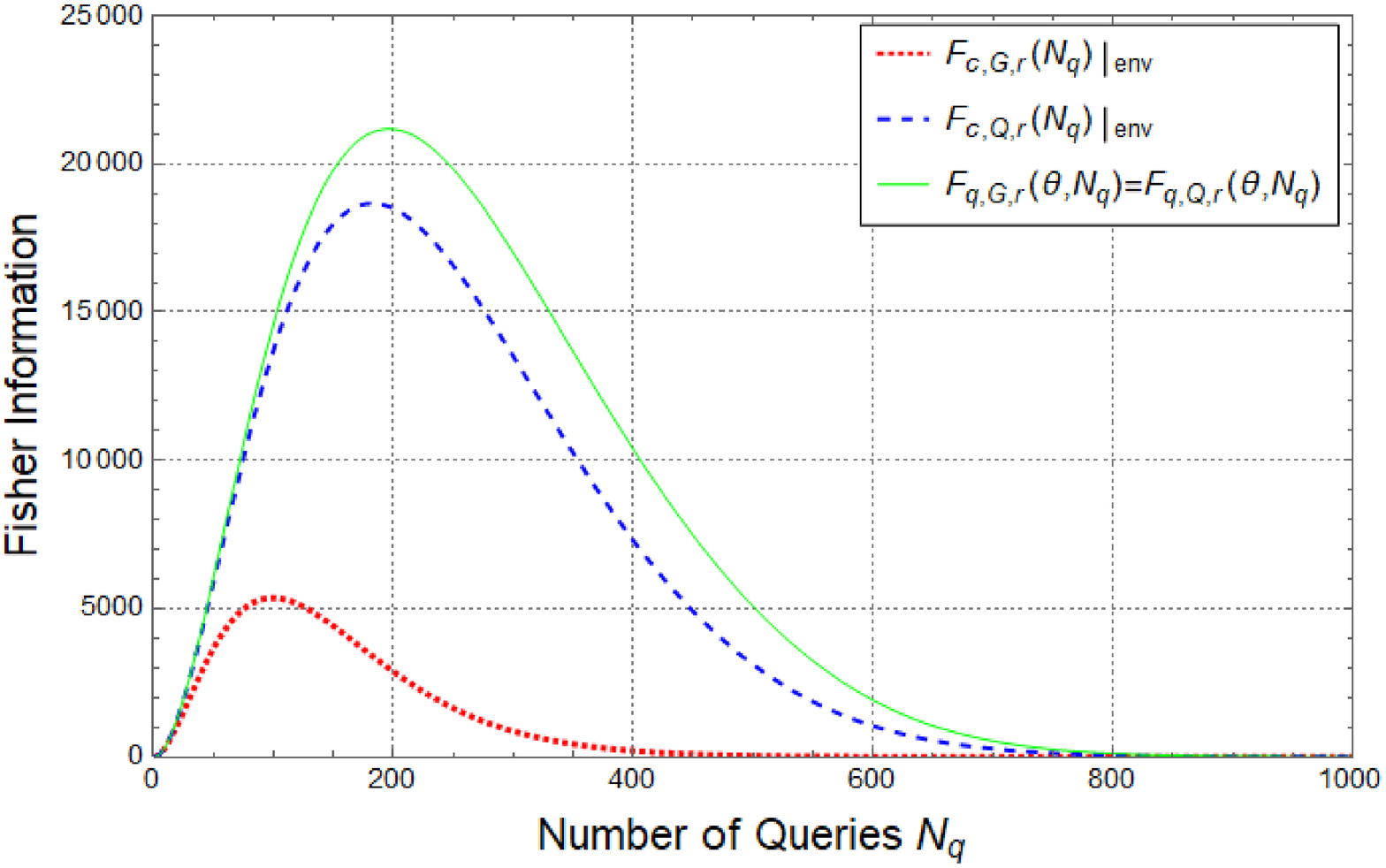}}
		\end{minipage}
	\end{minipage}
	\begin{minipage}[b]{0.5\linewidth}
		\begin{minipage}[b]{\linewidth}
			\centering
			\subfigure[$n=100\ (d = 2^{100})$]{\includegraphics[width=\linewidth]{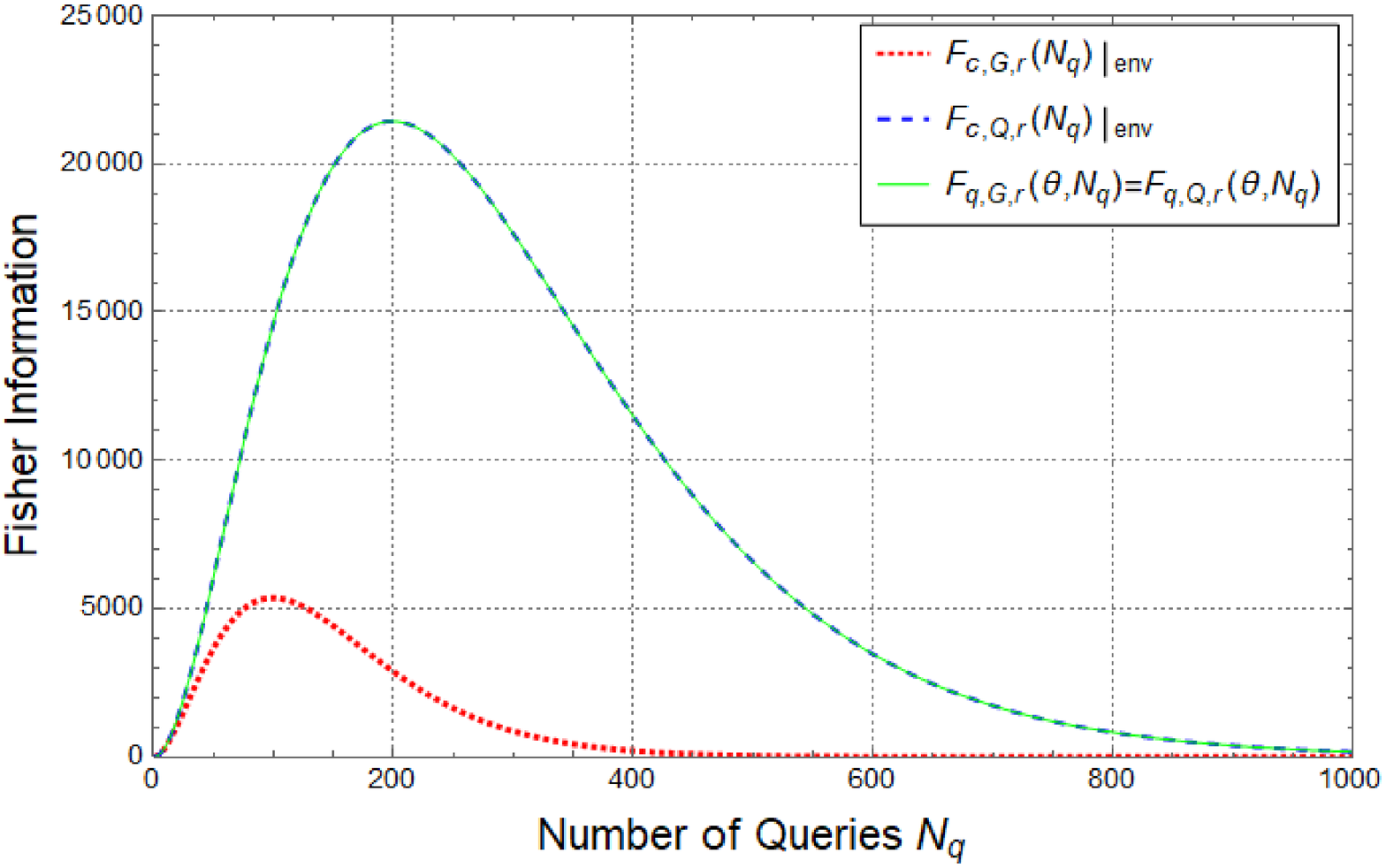}}
		\end{minipage}
		\\\\\\
		\begin{minipage}[b]{\linewidth}
			\centering
			\subfigure[$n\to\infty\ (d\to\infty)$]{\includegraphics[width=\linewidth]{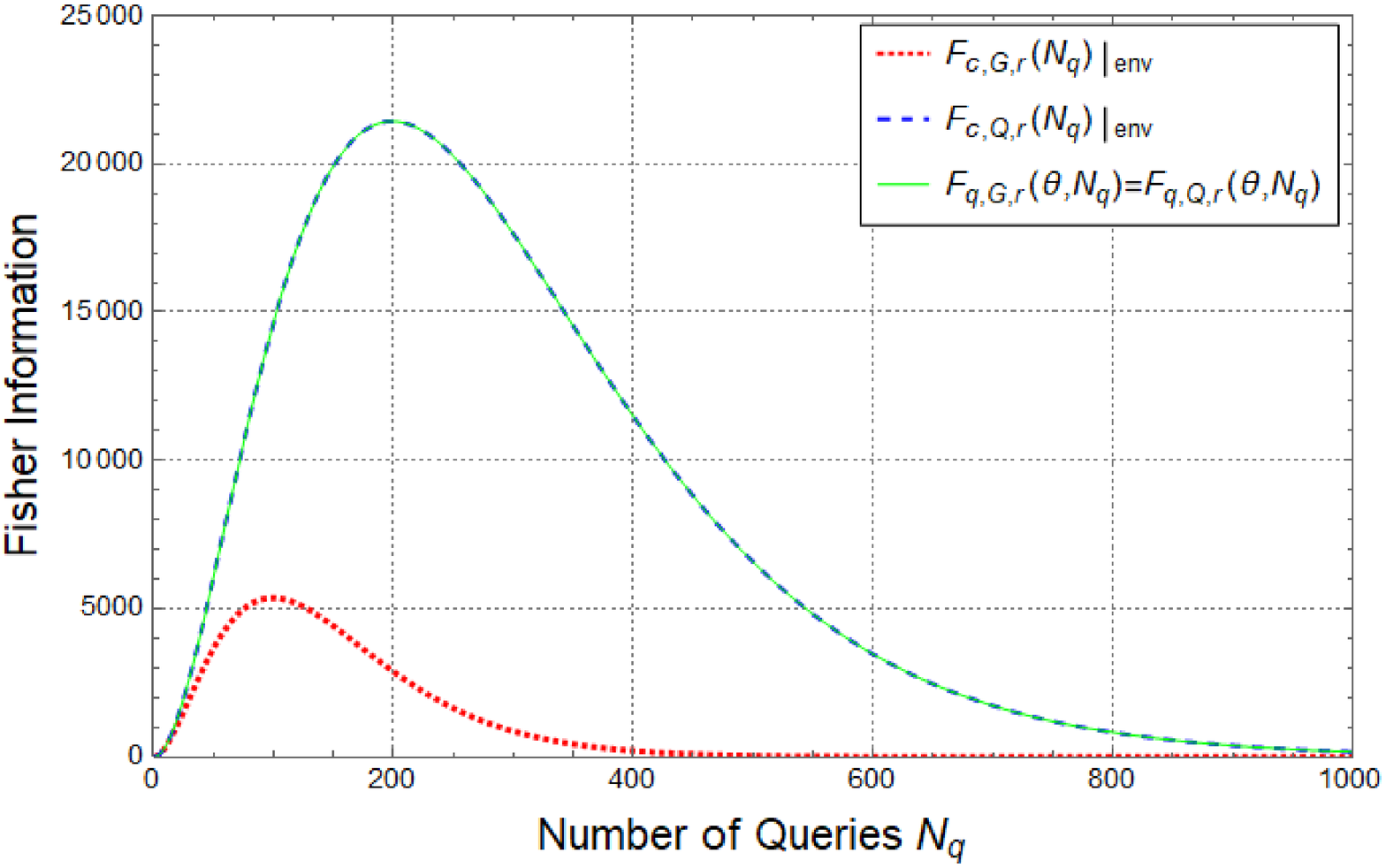}}
		\end{minipage}
	\end{minipage}
	\caption{
	Relationship between the number of queries and the Fisher information, with several values of the number of qubits, $n$.
	The envelope of the classical Fisher information of the conventional $G$-based method, $F_{\textrm{c},G,r}(N_\textrm{q})|_{\textrm{env}}$ defined in \eqref{eq:CFisherNoiseG_Max}, is depicted as the (red) dotted line, and the envelope of the classical Fisher information of the proposed $Q$-based method,  $F_{\textrm{c},Q,r}(N_\textrm{q})|_{\textrm{env}}$ defined in \eqref{eq:CFisherNoiseQ_Max}, is depicted as the (blue) dashed line.
	The quantum Fisher information $F_{\textrm{q},Q,r}(\theta,N_\textrm{q})=F_{\textrm{q},G,r}(\theta,N_\textrm{q})$ defined in \eqref{eq:QFisherNoiseQ} and \eqref{eq:QFisherNoiseG} is depicted as the (green) solid line.
	The noise parameter is set as $r=0.99$.
	}\label{fig:FisherInformation_p99}
\end{figure}

Figure~\ref{fig:FisherInformation_p99} illustrates the classical and quantum Fisher information of the $G$- and $Q$-based methods,
as a function of number of queries $N_\textrm{q}$.
Figure~\ref{fig:FisherInformation_p99}(a) shows that all the Fisher information presented above coincides with each other for the single
qubit case $n=1$.
For the case (b) $n=10$, a clear difference between $F_{\textrm{c},G,r}(N_\textrm{q})|_{\textrm{env}}$ and $F_{\textrm{c},Q,r}(N_\textrm{q})|_{\textrm{env}}$ can be seen;
notably, the latter is already close to the ultimate bound $F_{\textrm{q},Q,r}(N_\textrm{q})$ even with such a small number of qubits.
For the case (c) $n=100$, $F_{\textrm{c},Q,r}(N_\textrm{q})|_{\textrm{env}}$ almost reaches the bound $F_{\textrm{q},Q,r}(N_\textrm{q})$, while $F_{\textrm{c},G,r}(N_\textrm{q})|_{\textrm{env}}$ does show nearly zero change
from that of (b).
The shape of functions in (c) is almost the same as those shown in the case (d) $n=\infty$,
where $F_{\textrm{c},Q,r}(N_\textrm{q})|_{\textrm{env}}=F_{\textrm{q},Q,r}(N_\textrm{q})$
takes the maximum value $16/(e^2 \ln^2 r)$ at $N_\textrm{q}=-2 / \ln r$, which is four times larger than the maximum value of $F_{\textrm{c},G,r}(N_\textrm{q})|_{\textrm{env}}$, which takes $4/(e^2 \ln^2 r)$
at $N_\textrm{q}=-1 / \ln r$.
While we described the case with specific noise level $r=0.99$, the same
argument holds for other noise levels.
That is, for the Fisher information given in Eqs.~\eqref{eq:CFisherNoiseG_Max}, \eqref{eq:QFisherNoiseG} and \eqref{eq:CFisherNoiseQ_Max},
replacing $r$ with $r^c$ ($c$ is a positive constant) yields the same Fisher
information, except that the vertical axis is scaled with $1/c$ and horizontal
axis is scaled with $1/c^2$.

Therefore, as for the enveloped function, the $Q$-based method
using the $Q$ operator gives an estimator that not only outperforms
the $G$-based method but also, relatively easily, approaches to the ultimate estimation bound given by the inverse of quantum Fisher information.

Here we add a comment that the state $\Ket{0}_{n+1}$ is the optimal input for $Q$, in contrast to $A\ket{0}_{n+1}$ being the optimal input for $G$ mentioned in Section~\ref{sec:DepolarizationG}; see~\ref{sec:ProofMaximumQFI}.
This indicates that the quantum Fisher information obtained here is the optimal value for any method using the operator $Q$.
Recall now that, in general, to achieve the estimation accuracy given
by the inverse of quantum Fisher information, we need to elaborate
a sophisticated adaptive measurement strategy, e.g., varying the
number of amplitude amplification operations depending on
the measurement result.
In the next section, nonetheless, we demonstrate by numerical simulations that the $Q$-based method has a solid advantage over the $G$-based method even when employing a simple, non-adaptive measurement method in which a sequence of amplitude amplification operators is scheduled in advance.

Lastly we discuss on the readout error.
Recall that the $Q$-based method requires all qubits to be measured, while the
$G$-based method requires only one qubit measurement.
This may seem to cause large readout error in the $Q$-based method.
In fact, if the readout errors happen independently at every bit with
probability $\epsilon$, then in the $Q$-based method we will obtain a factor
of $(1-\epsilon)^n$ on the Fisher information.
The error rate of the current devices is about $\epsilon = 0.01$.
In this case the advantage of the $Q$-based method against the $G$-based
method will diminish at around $n = 70$ qubits, which is obtained by solving
$(1-\epsilon)^n=1/2$.
Thus, to fully utilize the $Q$-based method, we need to apply the readout error
mitigation technique, such as \cite{bravyi2021mitigating}.
Note that those techniques in general require exponential-time post processing,
but in our case we can overcome the limitations by exploiting the unique
property of the $Q$-based method.
That is, in order to obtain a precise estimate of $\theta$, it only needs
to correctly estimate the probability of observing the all-zero bits, in
contrast to the $G$-based method that needs to correctly estimate the
probability of observing bits whose last one is zero (there can be up to
$2^{n-1}$ of them).
Hence, when we perform $N_{shot}$ measurements and obtain at most
$N_{shot}\ll 2^n$ different outcomes, we only need to fix the frequency of
observing the all-zero bits from those $N_{shot}$ outcomes in the $Q$-based
method; this can be done efficiently (polynomial in $n$ and $N_{shot}$) for
the case of tensor product error model\cite{yang2021efficient}.

\section{Numerical Simulation}\label{sec:NumericalSimulation}

The purpose of this simulation is to study the performance of an  estimator constructed via the measurement result, compared to the
inverse of Fisher information.
Here we employ the same maximum likelihood estimator developed in
the previous papers~\cite{suzuki2020amplitude,tanaka2020amplitude},
which is non-adaptive and thus is not guaranteed to achieve the
$\theta$-dependent quantum Fisher information.
This estimation method begins with preparing several states in
parallel, where the amplitude of the target state is amplified
according to a pre-scheduled sequence $\{m_k\}$ for $k=0,1,\ldots$;
then we make measurements on these states characterized by $m_k$
and finally combine all the measurement results to construct the
maximum likelihood estimator on the amplitude parameter.
In particular, to achieve the Heisenberg scaling in the
region where the noise has little effect on the measurement results,
we employ an exponentially increasing sequence
$m_k=\lfloor b^{k-1} \rfloor$ for $k=0,1,\ldots$, where $b$ is some
real number greater than $1$; in this work, we take $b=6/5$.
Also, we fix the number of measurement to $N_{\textrm{shot}}=100$
for all $k$.
The number of qubit is $n=100$ (i.e., $d=2^{100}$) and the noise
strength is $r=0.99$ as in the case (c) in Figure~\ref{fig:FisherInformation_p99}.
The parameter values used in the numerical simulation are listed in Table~\ref{tab:ExperimentalParamters}.
We repeated the same experiment $200$ times to evaluate the root mean squared error of the estimate $\hat{\theta}$.
Since the proposed $Q$-based method with $m_0=0$ does not depend on
the parameter~$\theta$, the measurement result obtained for the case
$m_0=0$ is not used for the $Q$ case, but note that it is used for the conventional $G$ case.
Also recall that the depolarization process~\eqref{eq:DepolarizationChannel} acts on the state every time when
$A$ is operated.

Figure~\ref{fig:CRB_p999} shows the relationship between the root mean squared estimation error and the total number of queries defined
as $N_\textrm{q}^{\textrm{tot}} = \sum_{k}N_{\textrm{shot}}N_\textrm{q}(m_k)$, for several values of
target amplitude $a=\sin^2\theta$.
In each subfigure, the simulation results using the maximum likelihood estimator described above are plotted by the circle and triangle
points for the $G$ and $Q$ cases respectively.
Also the three lines depict the quantum and classical Cram\'er--Rao
lower bounds given by the inverse of Fisher information.
In addition, for reference, the gray dot-dashed line and the orange dot-dot-dashed line depict the lower bound in the ideal noiseless case (i.e., the Heisenberg scaling) and that without any amplitude amplification (i.e., the standard quantum limit), respectively.
We then see that, for all the cases from (a) to (f) of Figure~\ref{fig:CRB_p999}, the almost all simulation results well approximate the classical Cram\'er--Rao lower bound (CCRB), as expected from the asymptotic efficiency of maximum likelihood estimate.
The deviation in the subfigures (a), (b), and (c) for a small range of $N_q^{tot}$ would be caused by the bias of maximum likelihood estimate. We have observed that this deviation decreases when the number of shots is further increased, which is consistent with the fact that the maximum likelihood estimate is asymptotically unbiased in the limit $N_\textrm{shot}\to \infty$.
Except for these few biased results, in the short range of $N_\textrm{q}^{\textrm{tot}}$ where the influence from the noise is small, the $N_\textrm{q}^{\textrm{tot}}$-dependence of the estimation error of the $Q$-based method is almost identical with the $G$-based method.
Moreover, in this region, they are also almost identical with the error given by the quantum Fisher information, all exhibiting the Heisenberg scaling.
More precisely, the $G$-based method shows the Heisenberg scaling up to about $ N_\textrm{q}^{\textrm{tot}}\sim 0.5\times 10^5$, while the $Q$-based method does up to $ N_\textrm{q}^{\textrm{tot}}\sim 1\times 10^5$.
As $N_\textrm{q}$ increases, the estimation errors saturate to constant values due to the influence of noise, where there is a difference between $Q$ and $G$; that is, the saturated error of the former is about half  the latter.
Also, the saturated error of the $G$ case is about $2.5$ to $3$ times larger than the quantum Fisher information, whereas that of $Q$ is
about $1.2$ to $1.5$ times larger than the quantum Fisher information.
Around the rightmost points of the graphs, which correspond to $m\sim700$ amplifications, the lines without amplitude amplification intersects with those of the $G$-based one; but beyond this range the $Q$-based algorithm still has a superiority.
Hence there is certainly a range of queries where the $Q$-based algorithm is better than the $G$-based one as well as the one without amplifications even relatively large noise $1-r=0.01$
(i.e., 1$\%$ depolarization noise).
Overall, these results may be used as a numerical evidence to emphasize that the $Q$-based method, even without an adaptive setting, has a practical advantage over the $G$-based one.

\begin{table}[htb]
	\begin{center}
		\caption{List of parameters for the numerical simulation}
		\begin{tabular}{cll} \hline
			number of measurements     & $N_{\textrm{shot}}$     & 100                              \\
			amplification rule         & $m_k, (k=0,1,2,\ldots)$ & $\lfloor (6/5)^{k-1} \rfloor$    \\
			noise parameter            & $r$                     & $0.99$                           \\
			dimension of Hilbert space & $d$                     & $2^{100}$                        \\
			target values              & $a = \sin^2\theta$      & $\{2/3,1/3,1/6,1/12,1/24,1/48\}$ \\
			\hline
		\end{tabular}
		\label{tab:ExperimentalParamters}
	\end{center}
\end{table}

\begin{figure}[hbtp]

	\begin{minipage}[b]{0.5\linewidth}
		\begin{minipage}[b]{\linewidth}
			\centering
			\subfigure[$a = 2/3$]{\includegraphics[width=\linewidth]{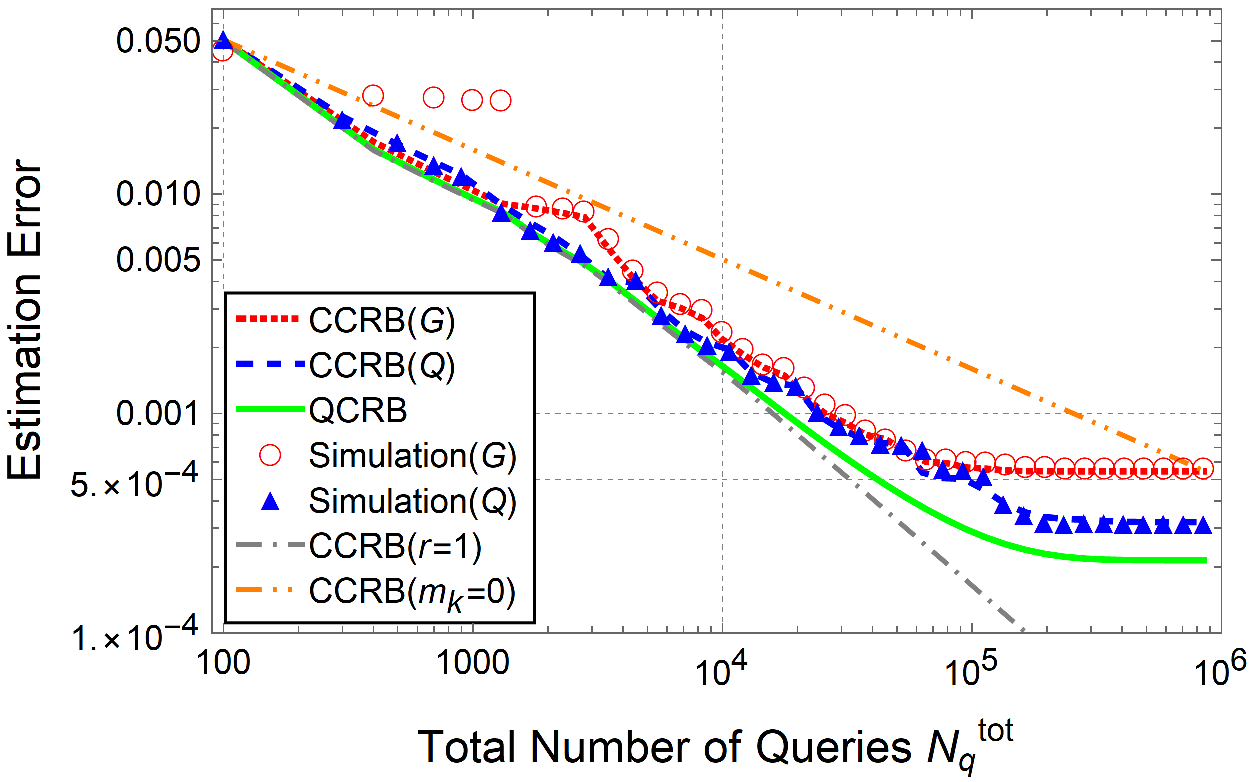}}
		\end{minipage} \\
		\begin{minipage}[b]{\linewidth}
			\centering
			\subfigure[$a = 1/3$]{\includegraphics[width=\linewidth]{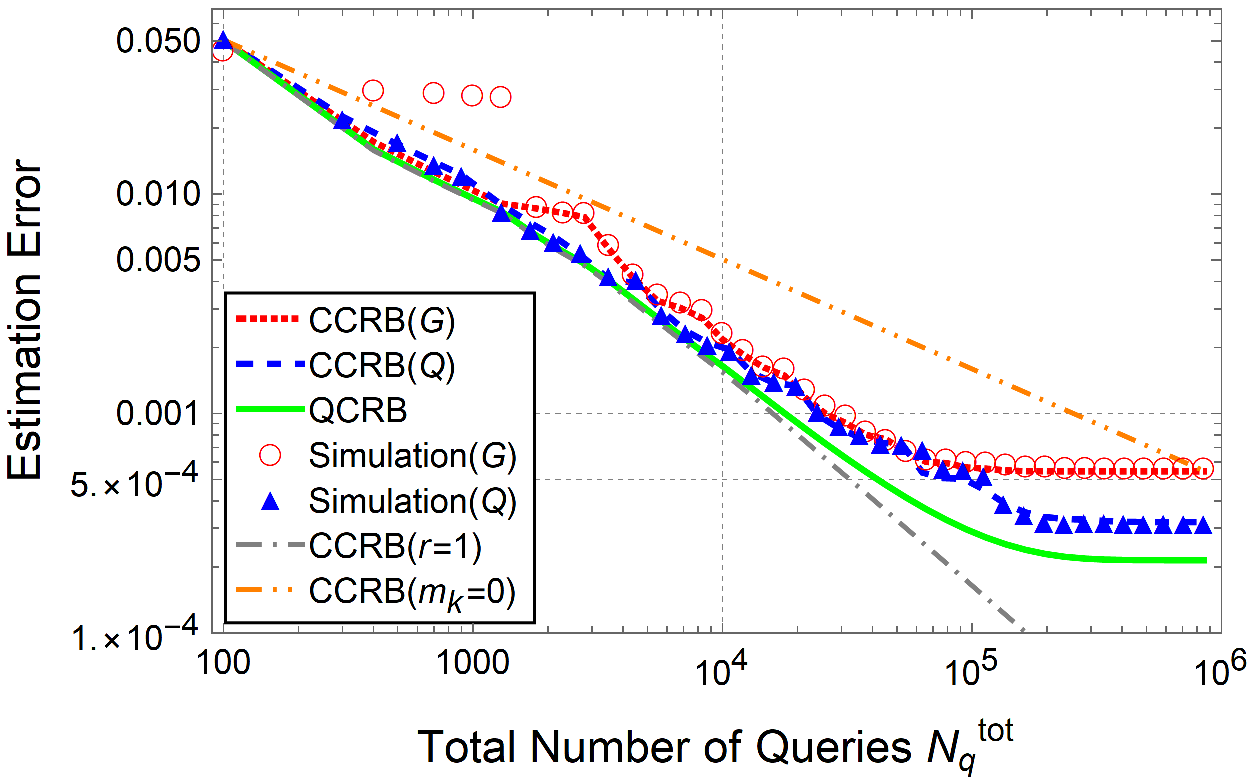}}
		\end{minipage} \\
		\begin{minipage}[b]{\linewidth}
			\centering
			\subfigure[$a = 1/6$]{\includegraphics[width=\linewidth]{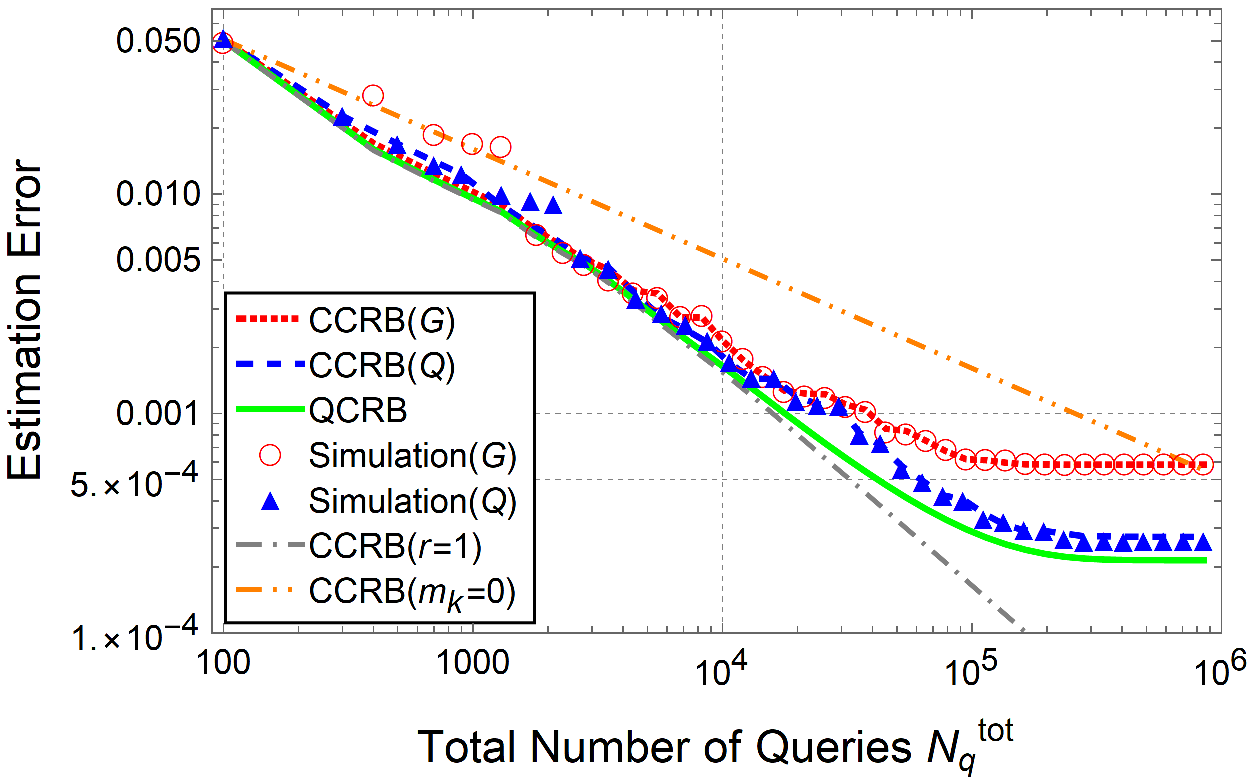}}
		\end{minipage}
	\end{minipage}
	\begin{minipage}[b]{0.5\linewidth}
		\begin{minipage}[b]{\linewidth}
			\centering
			\subfigure[$a = 1/12$]{\includegraphics[width=\linewidth]{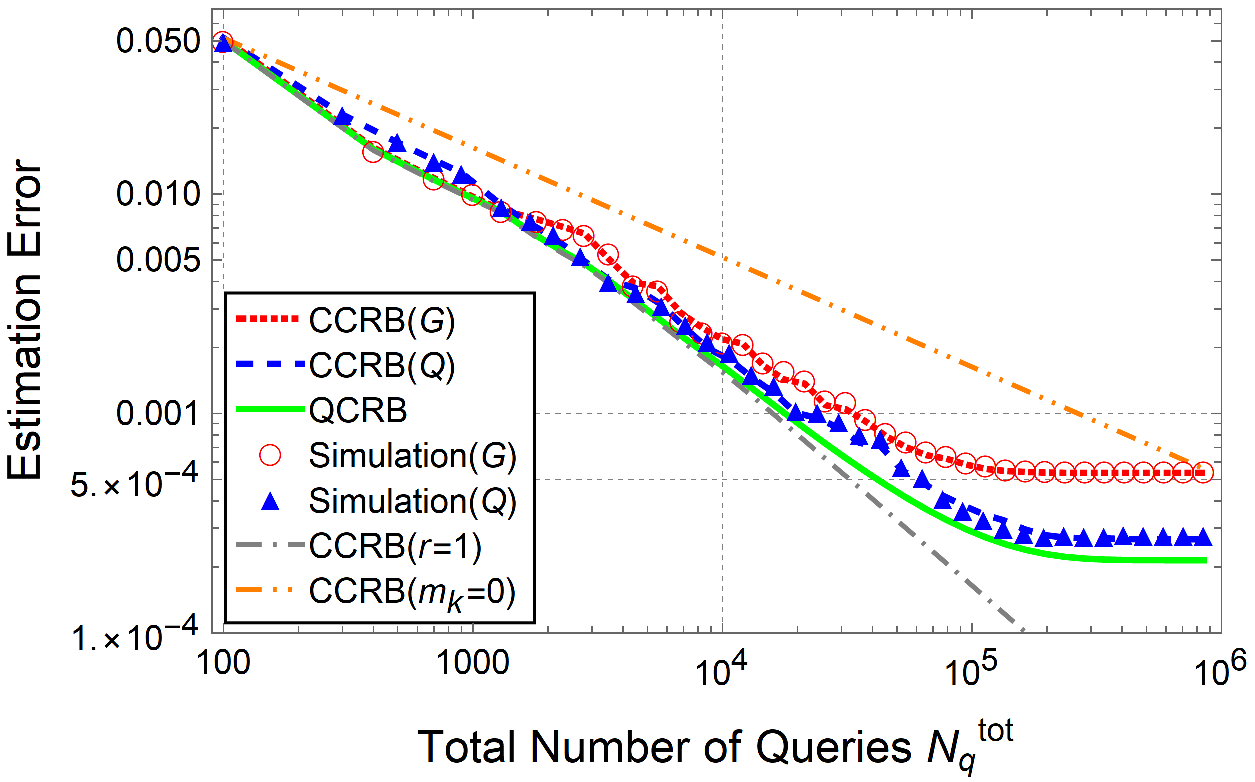}}
		\end{minipage}
		\begin{minipage}[b]{\linewidth}
			\centering
			\subfigure[$a = 1/24$]{\includegraphics[width=\linewidth]{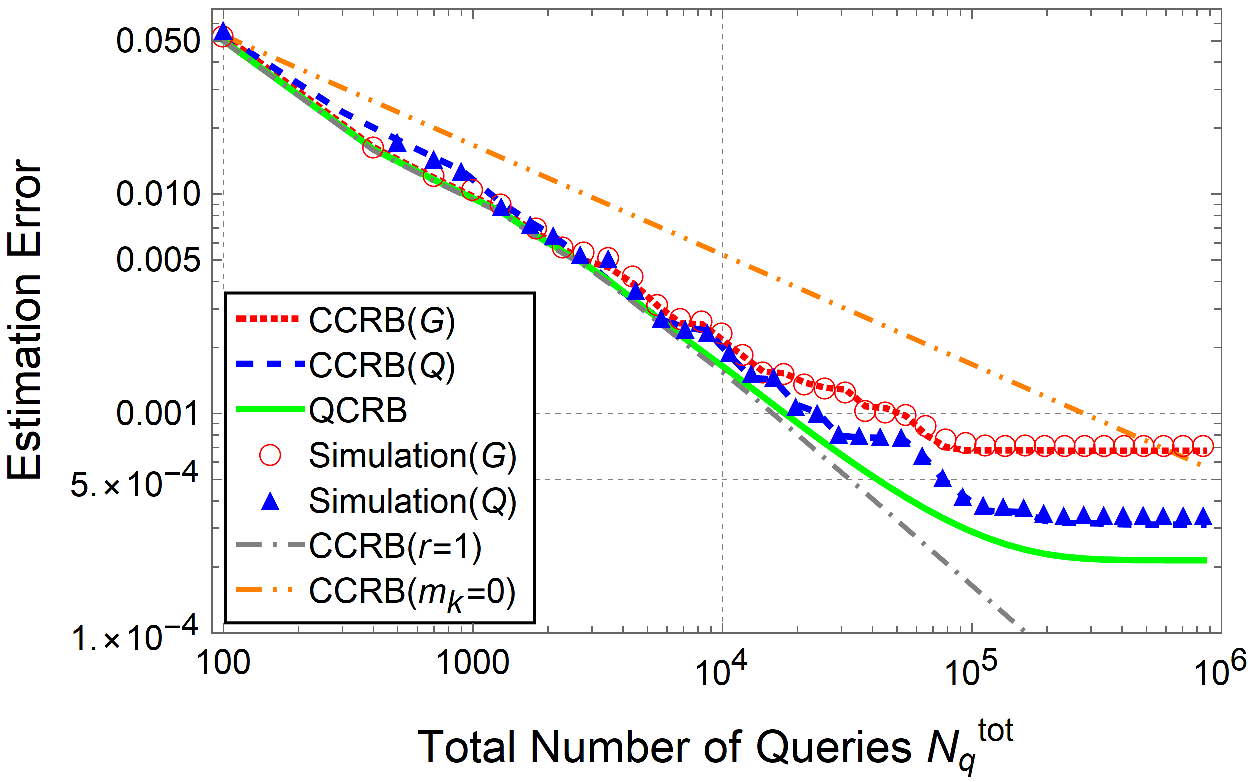}}
		\end{minipage} \\
		\begin{minipage}[b]{\linewidth}
			\centering
			\subfigure[$a = 1/48$]{\includegraphics[width=\linewidth]{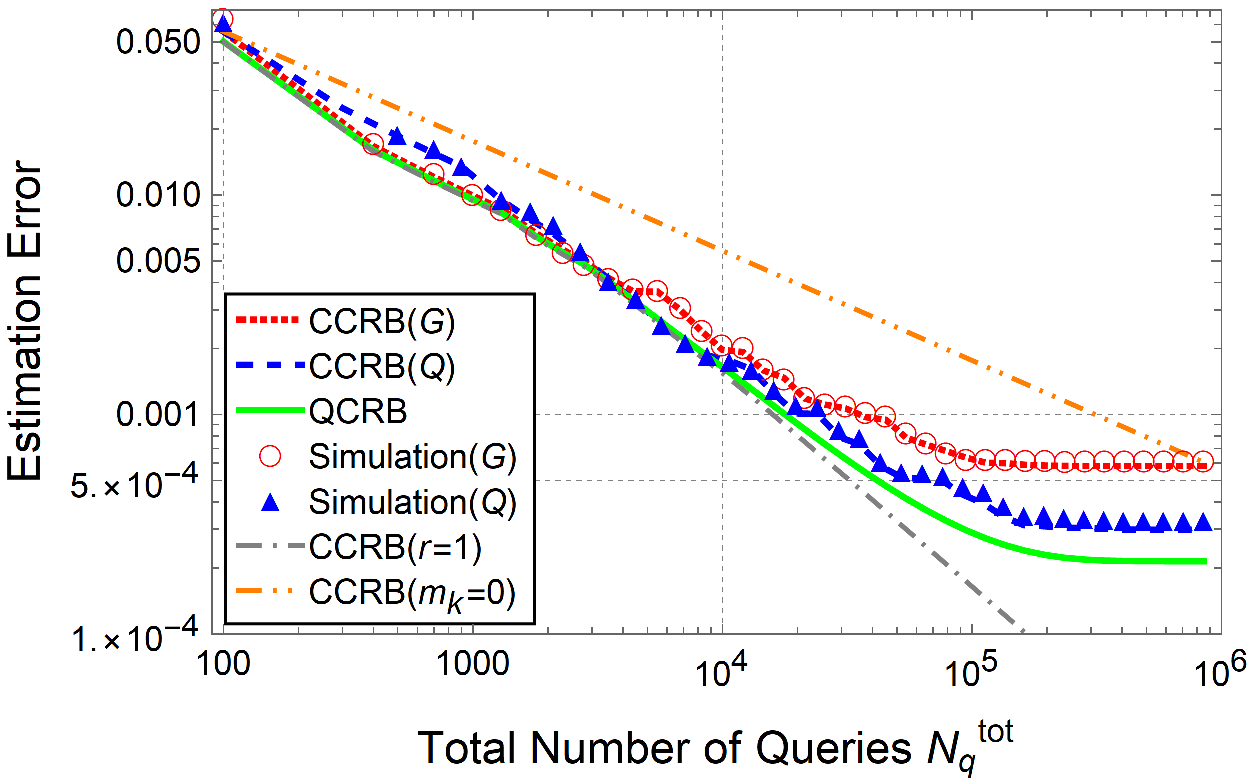}}
		\end{minipage}
	\end{minipage}
	\caption{
	Relationship between the total number of queries~$N_\textrm{q}^{\textrm{tot}}$ and the root mean squared error.
	The simulation parameters are listed in Table~\ref{tab:ExperimentalParamters}.
	The root mean square error obtained by the maximum likelihood estimation using the conventional $G$-based and proposed $Q$-based methods are plotted as (red) circles and (blue) triangles, respectively.
	The lower bounds of the Cram\'er--Rao inequality obtained from the quantum Fisher information $F_{\textrm{q},Q,r}(\theta,N_\textrm{q})=F_{\textrm{q},G,r}(\theta,N_\textrm{q})$, the classical Fisher information of the $Q$-based method $F_{\textrm{c},Q,r}(\theta,N_\textrm{q})$, and the classical Fisher information of the conventional method $F_{\textrm{c},G,r}(\theta,N_\textrm{q})$ are shown by the (green) solid line, the (blue) broken line, and the (red) dotted line, respectively.
	For reference, the lower bounds without noise ($r=1$) and without amplitude amplification ($m_k=0$) are shown by the (gray) dot-dashed line and the (orange) dot-dot-dashed line, respectively.
	}\label{fig:CRB_p999}
\end{figure}

\section{Conclusion}
\label{sec:Conclusion}

In this paper, we proposed the amplitude estimation method exhibiting
the better estimation accuracy than the conventional $G$-based method under the depolarizing noise.
Firstly, we showed that there is no measurement strategy to achieve the ultimate estimation accuracy in the conventional $G$-based method, except the 1-qubit case.
In contrast, the proposed $Q$-based method produces a larger classical Fisher information than the $G$-based one (in the sense of an enveloped function) whenever the number of qubit is larger than one, and moreover, it achieves the ultimate estimation accuracy in the limit of large number of qubits.
Note that, achieving the ultimate estimation accuracy generally requires a sophisticated adaptive algorithm, but our numerical simulations confirmed that the estimator in the $Q$-based method
well approaches to the ultimate bound even with a simple
non-adaptive strategy.

In this paper, we have numerically confirmed the advantage of the $Q$-based method to the $G$-based one for a non-adaptive measurement strategy, but concrete adaptive algorithms to achieve the quantum Fisher information are open for future research.
In addition, although the depolarizing noise was analyzed in this study, it is also important to analyze other noises such as dephasing noise and measurement noise.
Also, in this paper, we assume that the noise parameter, the parameter $r$ in \eqref{eq:DepolarizationChannel}, is known, but it is also important to analyze the case where the parameter is unknown, as studied in the recent paper~\cite{tanaka2020amplitude}.

\section*{Acknowledgement}
This work was supported by MEXT Quantum Leap Flagship Program Grant Number JPMXS0118067285 and JPMXS0120319794.

\afterpage{\clearpage}

\appendix
\section{Proof of the optimality of the quantum Fisher information}\label{sec:ProofMaximumQFI}

In the main text, we presented the quantum Fisher information in the absence of noise, \eqref{eq:QFisherPureG} and \eqref{eq:QFisherPureQ}, and those under the depolarizing noise, \eqref{eq:QFisherNoiseG} and \eqref{eq:QFisherNoiseQ}.
In this Appendix, we prove that these quantum Fisher information are optimal when employing a sequential strategy.

The main theorem is as follows.
\begin{theorem}
	Let $\ket{\psi_N}$ be a state defined in a $d$-dimensional Hilbert space as
	\begin{equation}
		\Ket{\psi_N} = U_N V_N(\theta) U_{N-1} V_{N-1}(\theta) \cdots U_2 V_2(\theta) U_1 V_1(\theta) U_0 \Ket{0}.
		\label{eq:StateNTimesOracle}
	\end{equation}
	where $U_i\ (i=0,1,\ldots,N)$ represents an arbitrary unitary operator independent of $\theta$,
	and $V_i(\theta)\ (i=1,\ldots,N)$ is a unitary operator which is differentiable with respect to $\theta$ and whose derivative is non-expansive (that is, $\abs{\frac{\partial V_i(\theta)}{\partial \theta} \Ket{v}} \le \abs{\ket{v}}$ holds for any vector $\Ket{v}$). When we estimate $\theta$ from the measurement results sampled from the state $\ket{\psi_N}$, the quantum Fisher information $F$ satisfies the following inequalities:
	\begin{enumerate}
		\item{$F\le 4 N^2$ in the absence of noise.}
		\item{ $ F\le \frac{4 {N}^2 \left( \prod_{i=1}^N r_i \right)^2}{\frac{2}{d}+ \left( 1-\frac{2}{d} \right)  \prod_{i=1}^N r_i}$ under the depolarizing noise, where $V_i(\theta)\ (i=1,2,\ldots,N)$ is subject to the depolarizing noise of strength $r_i$.
		      }
	\end{enumerate}
	\label{the:UpperBoundQuantumFisher}
\end{theorem}

The operators $A$ and $A^\dagger$ of \eqref{eq:InitialStateG} can be regarded as examples of $V_i(\theta)$ in this theorem, if the action of $\frac{\partial A}{\partial \theta}$ on states other than $\Ket{0}_{n+1}$ and $\Ket{\phi}_{n+1}$ (and $\frac{\partial A^\dagger}{\partial \theta}$ on states other than $\Ket{\psi_0}_n\Ket{0}$ and $\Ket{\psi_1}_n\Ket{1}$) is non-expansive.
Therefore, we can conclude that the quantum Fisher information obtained in the main text is optimal in the context of Theorem~\ref{the:UpperBoundQuantumFisher}.

In the following, we give the proof of
Theorem~\ref{the:UpperBoundQuantumFisher}, based on the following three Lemmas.

\begin{lemma}
	$\braket{\dot{\psi}_N|\psi_N}$ is a pure imaginary number. Here the overdot represents the derivative with respect to $\theta$.
	\label{lem:innerProducPsiPsiDot}
\end{lemma}
\begin{proof}
	Differentiating both sides of $\braket{\psi_N|\psi_N}=1$ with $\theta$ gives the equation $\braket{\dot{\psi}_N|\psi_N} = -\braket{\dot{\psi}_N|\psi_N}^*$.
\end{proof}
\begin{lemma}
	The inequality $\braket{v_1|\dot{V}_i(\theta)|v_2} + \braket{v_1|\dot{V}_i(\theta)|v_2}^*\le 2$ holds for any normalized vector $\Ket{v_1}$ and $\Ket{v_2}$.
	The equality holds only when $\braket{v_1|\dot{V}_i(\theta)|v_2} = 1$.
	\label{lem:InnerProduct}
\end{lemma}
\begin{proof}
	It is straightforward from the assumption that the operator $\dot{V}_i(\theta)$ is non-expansive.
\end{proof}
\begin{lemma}
	The norm $\abs{\ket{\dot{\psi}_N}}$ has an upper bound $\abs{\ket{\dot{\psi}_N}}^2\le N^2$.
	\label{lem:normPsiDot}
\end{lemma}
\begin{proof}
	We can expand the norm as
	\begin{equation}
		\begin{split}
			\abs{\ket{\dot{\psi}_N}}^2
			= \sum_{i,j=1}^N\Bra{0}_{n+1} U_0^\dagger& V_1^\dagger(\theta) \cdots U_{i-1}^\dagger \dot{V}_i^\dagger(\theta) U_i^\dagger\cdots V^\dagger_N(\theta) U_N^\dagger  \\
			&U_N V_N(\theta) \cdots U_j \dot{V}_j(\theta) U_{j-1}\cdots V_1(\theta) U_0\Ket{0}_{n+1}.
			\label{eq:normPsiDot}
		\end{split}
	\end{equation}
	The right hand side of this equation takes the maximum value only when each term in the summation take the value~$1$, which is derived from Lemma~\ref{lem:InnerProduct} for $i\ne j$ and from the assumption that $V_i(\theta)$ is a non-expansive map for $i=j$.
\end{proof}
From these lemmas, the first part of the Theorem~\ref{the:UpperBoundQuantumFisher} (in the absence of noise) is proved as follows.
\begin{proof}
	Recall that the quantum Fisher information for a pure state $\Ket{\psi_N}$ is obtained as
	\begin{equation}
		\begin{split}
			F
			&= 4\left( \braket{\dot{\psi}_N|\dot{\psi}_N}-\left|\braket{\dot{\psi}_N|\psi_N}\right|^2 \right).
		\end{split}
	\end{equation}
	The first term in the parenthesis of the right hand side is upper bounded by $N^2$ from Lemma~\ref{lem:normPsiDot}, and the second term is upper bounded by $0$. Thus, we obtain the upper bound of quantum Fisher information as $F\le 4 N^2$.
	The equality holds only when $\braket{\dot{\psi}_N|\dot{\psi}_N}=N^2$ and $\braket{\dot{\psi}_N|\psi_N}=0$.
\end{proof}
Also, the second part of the Theorem~\ref{the:UpperBoundQuantumFisher} (under the depolarizing noise) is proved as follows.
\begin{proof}
	Assuming that the depolarizing noise with strength $r_i$ of the form \eqref{eq:DepolarizationChannel} acts after the operation of $V_i(\theta)$, the state of \eqref{eq:StateNTimesOracle} becomes
	\begin{equation}
		\rho_N^r = \tilde{r} \rho_N + \left(1-\tilde{r}\right) \frac{\mathbf{I}_{n+1}} {d}.
	\end{equation}
	Here we denote $\rho_N = \Ket{\psi_N}\Bra{\psi_N}$ and $\tilde{r}=\prod_{i=1}^N r_i$.
	Following the previous papers\cite{jiang2014quantum,yao2014multiple}, SLD operator of this state can be written as
	\begin{equation}
		L_S = \frac{2  \tilde{r}}{\frac{2}{d} + \left(1-\frac{2}{d}\right)\tilde{r}} \dot{\rho}_N.
	\end{equation}
	Then the quantum Fisher information can be calculated as
	\begin{equation}
		\begin{split}
			F &= \mathrm{Tr}\left( \rho_N^r L_S^2 \right) \\
			&=\left(\frac{2 \tilde{r} }{\frac{2}{d} + \left(1-\frac{2}{d}\right)\tilde{r}}\right)^2
			\left[
				\tilde{r}
				\left( \braket{\dot{\psi}_N| \psi_N}^2 +  \braket{\psi_N | \dot{\psi}_N }^2
				+  \braket{\dot{\psi}_N | \dot{\psi}_N }
				+  \abs{\braket{\psi_N | \dot{\psi}_N }}^2
				\right)					\right. \\
				&\left.
				\hspace{5cm} + \frac{ 1-\tilde{r}}{d}
				\left(
				\braket{\dot{\psi}_N| \psi_N}^2 +  \braket{\psi_N | \dot{\psi}_N }^2
				+  2 \braket{\dot{\psi}_N | \dot{\psi}_N }
				\right)					\right].
			\label{eq:QFisherArbitrary}
		\end{split}
	\end{equation}
	The first term in the bracket of the right hand side takes the maximum only when $\braket{\psi_N | \dot{\psi}_N } =0 $ and $\braket{\dot{\psi}_N | \dot{\psi}_N } =N^2 $ hold from Lemmas~\ref{lem:innerProducPsiPsiDot} and \ref{lem:normPsiDot}.
	The second term also takes the maximum under the same conditions.
	Substituting these conditions into \eqref{eq:QFisherArbitrary}, we obtain the upper bound of the quantum Fisher information as $F\le \frac{4 {N}^2 \tilde{r}^{2 }}{\frac{2}{d}+ \left( 1-\frac{2}{d} \right)\tilde{r} }$.
	The equality holds only when $\braket{\dot{\psi}_N|\dot{\psi}_N}=N^2$ and $\braket{\dot{\psi}_N|\psi_N}=0$.
\end{proof}

\section*{References}

\bibliographystyle{iopart-num}
\bibliography{main}
\end{document}